\documentclass[11pt, letterpaper]{article}
\usepackage[margin=1in]{geometry}

\usepackage[dvipsnames]{xcolor}

\usepackage{xfrac}
\usepackage{enumitem}
\usepackage{pifont}
\usepackage{colortbl}
\usepackage{hhline}
\usepackage{float}
\usepackage{bm}
\usepackage{amsmath,amssymb}
\usepackage{thm-restate}
\usepackage[colorlinks=true,citecolor=blue,linkcolor=blue,urlcolor=black]{hyperref}
\usepackage{amsthm}

\newtheorem{theorem}{Theorem}[section]
\newtheorem{lemma}[theorem]{Lemma}
\newtheorem{proposition}[theorem]{Proposition}

\theoremstyle{definition}
\newtheorem{definition}[theorem]{Definition}

\DeclareMathOperator*{\argmax}{arg\,max}
\DeclareMathOperator*{\argmin}{arg\,min}
\newcommand{\cmark}{\cellcolor{green!20}{\ding{51}}}
\newcommand{\xmark}{\cellcolor{red!20}{\ding{55}}}
\usepackage{multirow}
\usepackage{array}
\newcolumntype{P}[1]{>{\centering\arraybackslash}m{#1}}
\usepackage{tikz}
\newcommand{\sw}{\textsc{SW}}

\usepackage{authblk}
\author[1]{Pooja Kulkarni}
\author[2]{Ruta Mehta}
\author[2]{Vishnu V. Narayan}
\author[3]{Tomasz Ponitka}
\affil[1]{Northwestern University}
\affil[2]{University of Illinois Urbana-Champaign}
\affil[3]{Tel Aviv University}
\usepackage[numbers,sort]{natbib}
\affil[ ]{\footnotesize\texttt{pooja.kulkarni@northwestern.edu, rutameht@illinois.edu, vishnu.narayan@mail.mcgill.ca, tomaszp@mail.tau.ac.il}}

\newcommand\blfootnote[1]{
  \begingroup
  \renewcommand\thefootnote{}
  \NoHyper\footnote{#1}\endNoHyper
  \addtocounter{footnote}{-1}
  \endgroup
}

\title{Online Fair Division With Subsidy: \\When Do Envy-Free Allocations Exist, and at What Cost?}
\date{October 15, 2025}

\newcommand{\BibTeX}{\rm B\kern-.05em{\sc i\kern-.025em b}\kern-.08em\TeX}

\allowdisplaybreaks

\begin{document}

\maketitle 

\blfootnote{\hspace*{-2.2em}
A portion of this work was completed while P. Kulkarni was a student at the University of Illinois Urbana-Champaign. P. Kulkarni, R. Mehta and V. V. Narayan were supported by NSF grant CCF-2334461. P. Kulkarni was supported by a gift from Adobe to S. Khuller. T. Ponitka was supported by ERC grant 101170373, an Amazon Research Award, NSF-BSF grant 2020788, Israel Science Foundation grant 2600/24, and a TAU Center for AI and Data Science grant.}

\begin{abstract}
We study the problem of fairly allocating $m$ indivisible items arriving online, among $n$ (offline) agents. Although envy-freeness has emerged as the archetypal fairness notion,  envy-free (EF) allocations need not exist with indivisible items. To bypass this, a prominent line of research demonstrates that there exist allocations that can be made envy-free by allowing a {\em subsidy}. Extensive work in the offline setting has focused on finding such {\em envy-freeable} allocations with bounded subsidy. We extend this literature to an online setting where items arrive one at a time and must be immediately and irrevocably allocated.  Our contributions are two-fold:
\begin{itemize}[leftmargin=*]
    \item \textbf{Maintaining EF Online:} We show that envy-freeability cannot always be preserved online when the valuations are submodular or supermodular, even with binary marginals. In contrast, we design online algorithms that maintain envy-freeability at every step for the class of additive valuations, and for its superclasses including $k$-demand and SPLC valuations.
    \item \textbf{Ensuring Low Subsidy:} We investigate the quantity of subsidy required to guarantee envy-freeness online. Surprisingly, even for additive valuations, the minimum subsidy may be as large as $\Omega(mn)$,  in contrast to the offline setting, where the bound is $O(n)$. On the positive side, we identify valuation classes where the minimum subsidy is small (i.e., does not depend on $m$), including $k$-valued, rank-one, restricted additive, and identical valuations, and we obtain (mostly) tight subsidy bounds for these classes.
\end{itemize}
\end{abstract}

\section{Introduction}

In the decades since its formal introduction by \citet*{steinhaus1948problem} in 1948, the fair division problem has emerged as a fundamental topic at the intersection of economics and computer science. The central fair division question asks how to allocate a finite collection of items amongst a set of agents, each equipped with a valuation function over subsets of the items, in a manner that is universally regarded as fair. When the items are divisible, and with only mild assumptions on the structure of the agents’ valuations, it is possible to obtain an allocation where no agent is envious of any other, i.e., every agent (according to its own valuation) weakly prefers its assigned bundle to that of any other agent \cite{stromquist1980how}. This concept, introduced by \citet*{foley1967resource}, is now widely known as {\em envy-freeness}. However, when the items are indivisible and must each be integrally assigned to some agent, it is easy to see that envy-free allocations do not necessarily exist: when only one valuable item is available, the agent that receives it is always envied.

In a seminal work, \citet*{Maskin1987} asked whether this impossibility could be circumvented by the addition of a single divisible good (i.e., money), and if so, how much of it was needed to eliminate all envy. \citet*{Maskin1987} showed that for matching instances with $n$ agents and $m=n$ goods, if without loss of generality any specific good has a value of at most one dollar, an {\em envy-freeable} allocation always exists, and a subsidy of $n-1$ dollars is sufficient to remove all envy for some allocation. This result spurred a large body of literature investigating the existence of such envy-freeable allocations in fair division \cite{aragones1995derivation,Haake2002,halpern2019fair,BrustleDNSV20,Hartline2008}, leading to a now well-known characterization:

\textsc{Theorem} (\cite{aragones1995derivation,Haake2002,Hartline2008,halpern2019fair}). {\em An allocation is envy-freeable if and only if it is locally efficient}.

An allocation is locally efficient (LE) if, when the $n$ assigned bundles of the allocation are each frozen, no reassignment of these bundles to the $n$ agents increases the social welfare. Several papers have since examined the quantity of subsidy sufficient to guarantee envy-freeness. Prominently, the work of \citet*{BrustleDNSV20} proves that it is always possible to eliminate envy with a {\em small} total subsidy, i.e., one that does not grow with the number of items $m$.

We extend the above body of literature to a natural online setting in which items arrive one at a time, revealing their values, and must each be assigned immediately and irrevocably. Our interest in this problem arises from the simple observation that when the valuations are additive, it is possible to {\em maximize welfare} online (i.e., maintain a kind of {\em global} efficiency), implying the maintenance of LE and therefore envy-freeability online. Thus the first natural question that we raise and study is the following.

\begin{center}
\begin{tabular}{@{}l@{ }p{0.8\textwidth}@{}}
\textbf{Question 1.} & \textit{Going beyond additive valuations, for what valuation classes 
is it possible to maintain envy-freeability online?}
\end{tabular}
\end{center}

We show that it is additionally possible to maintain LE online for the classes of SPLC valuations and $k$-demand valuations (for any $k$). Interestingly, although welfare maximization is one possible strategy to obtain LE, the collection of valuation classes for which LE can be maintained online is strictly more general, i.e., this may be possible even when maximizing welfare online is impossible. In this sense, our work demonstrates a separation between classes for which online welfare maximization is possible and classes for which local efficiency can be maintained online.

In the other direction, we show that maintaining LE online is {\em impossible} for budget-additive valuations, even if the budgets are known in advance, and that it is impossible for submodular and supermodular valuations, even if the marginal values are binary.

While in the offline case, envy-freeness with a small subsidy is feasible for both additive valuations and monotone valuations, for the online problem, maintaining LE via welfare maximization results in a subsidy that grows linearly with $m$ in the worst case. This raises the second natural question that we study in this work.

\begin{center}
\begin{tabular}{@{}l@{ }p{0.8\textwidth}@{}}
\textbf{Question 2.} & \textit{Do there exist valuation classes for which it is possible to 
maintain envy-freeability online with a {\em small} subsidy, i.e., one that does not grow with $m$?}
\end{tabular}
\end{center}
We first show that when the agents have additive valuations this is impossible, and there exists an instance for which maintaining envy-freeability online necessitates a subsidy as high as $m(n-1)$. A similar lower bound extends to SPLC valuations and $k$-demand valuations (both of which generalize additive valuations when $k$ is allowed to be as large as $m$). 

In contrast, we provide a positive answer to the above question for several well-studied valuation classes, including unit-demand and bivalued valuations (and, more generally, $k$-demand and $k$-valued valuations when $k$ is independent of $m$), as well as rank-one, restricted additive, and identical monotone valuations. The definitions of these classes, along with a more detailed discussion of their properties, are presented in the respective sections.

All of our positive results are obtained algorithmically. Through adversarial constructions, we show that nearly all of our algorithms are \textit{exactly} tight, i.e., they have the best possible bound on the total subsidy, even up to coefficients. Interestingly, for several of these classes this bound is asymptotically larger (as a function of $n$) than the corresponding subsidy bound in the offline version of the problem, thereby providing strong separations between the online and offline settings.\footnote{For instance, a result of \citet*{BrustleDNSV20} shows that a total subsidy of $n-1$ suffices for additive valuations. Our results demonstrate that the minimum subsidy is $\Theta(n^2)$ even for structured subclasses of additive valuations, including binary additive valuations and rank-one valuations.} 

Table~\ref{tab:valuation-questions} shows a high-level summary of our results.

\newcommand{\MarkLeftMargin}{3.0em}
\newcommand{\MarkIndent}{4.5em} 

\newlength{\MarkRefTotalWidth}
\newlength{\MarkRefContentWidth}
\newcommand{\setmarkrefwidth}[1]{%
  \settowidth{\MarkRefTotalWidth}{\footnotesize (#1)}%
  \settowidth{\MarkRefContentWidth}{\footnotesize #1}%
}

\newcommand{\markonly}[1]{%
  \parbox[t]{\linewidth}{%
    \hangindent=\MarkIndent \hangafter=1 \noindent
    \hspace*{\MarkLeftMargin}%
    \makebox[0pt][l]{#1}%
    \hspace*{\dimexpr\MarkIndent-\MarkLeftMargin\relax}%
  }%
}

\newcommand{\markref}[2]{%
  \parbox[t]{\linewidth}{%
    \hangindent=\MarkIndent \hangafter=1 \noindent
    \hspace*{\MarkLeftMargin}%
    \makebox[0pt][l]{#1}%
    \hspace*{\dimexpr\MarkIndent-\MarkLeftMargin\relax}%
    \makebox[\MarkRefTotalWidth][l]{\footnotesize (#2)}%
  }%
}

\newcommand{\linefill}{\rule[0.5ex]{\MarkRefContentWidth}{0.4pt}}

\setmarkrefwidth{Thm \ref{thm:single_param}}

\begin{table}[t]
\centering
\begin{tabular}{|P{0.25\linewidth}|P{0.25\linewidth}|P{0.25\linewidth}|}
\hline
\multicolumn{1}{|c|}{\multirow{3}{*}{\textbf{Valuation Class}}} & \textbf{Question 1} & \textbf{Question 2} \\
 & {Can EF be maintained \textbf{online}?} &  {Can we get EF with \textbf{small} subsidy?} \\
\hline\hline
Budget Additive & \markref{\xmark}{Thm \ref{thm:budget_additive}} & \markref{\xmark}{\linefill}  \\
\hline
Binary Submodular & \markref{\xmark}{Thm \ref{thm:binary_submodular}} & \markref{\xmark}{\linefill} \\
\hline
Binary Supermodular & \markref{\xmark}{Thm \ref{thm:binary_supermodular}} & \markref{\xmark}{\linefill} \\
\hline
\hline
SPLC & \markref{\cmark}{Thm \ref{thm:splc}} & \markref{\xmark}{\linefill} \\
\hline
{Additive} & \markref{\cmark}{\linefill} & \markref{\xmark}{Thm \ref{thm:additive-lb}} \\
\hline\hline
$k$-Demand  & \markref{\cmark}{\linefill} & \markref{\cmark$^\ast$}{Thm \ref{thm:k-demand-subsidy}} \\
\hline
$k$-Valued & \markref{\cmark}{\linefill} & \markref{\cmark$^\ast$}{Thm \ref{thm:k-valued}}
\\
\hline
Rank-One & \markref{\cmark}{\linefill} &
\markref{\cmark}{Thm \ref{thm:single_param}} \\
\hline
Restricted Additive & \markref{\cmark}{\linefill} & \markref{\cmark}{Thm \ref{thm:restrictedadditive}} \\
\hline
Identical Monotone & \markref{\cmark}{\linefill} & \markref{\cmark}{Thm \ref{thm:identical_monotone}}  \\
\hline
\end{tabular}
\caption{Summary of our results for Questions 1 and 2 for the valuation classes considered in this work. $(^\ast)$ For $k$-valued and $k$-demand valuations, our subsidy bound is small, as it does not depend on $m$ but does depend on $k$.}
\label{tab:valuation-questions}
\end{table}

We remark that a key modeling choice in our work is to track the minimum subsidy required to guarantee envy-freeness at each step, rather than maintaining an explicit online update of the payment vector. At any point during the online arrival of items, the mechanism is prepared to obtain an envy-free allocation with subsidy should the process terminate. This choice aligns with applications in which the online arrival of items may halt unpredictably, and the envy-freeness guarantee must be immediately realized with bounded subsidy.

Finally, we note that most of our results extend to the related chore-division problem, where each item adds a \textit{cost} to the agent that receives it. Section~\ref{sec:discussion} contains a high-level discussion of the problem of bounding the subsidy in the fair division of chores.

\subsection{Additional Related Work}

While this work introduces the online fair division with subsidy problem with online items and offline agents, there is a large body of research on each of these topics taken individually. See the survey by \citet*{amanatidis2023fair} for more details on these topics.

\paragraph{Fair Division With Subsidy.} The problem of finding a fair allocation with subsidy has been considered in several contexts, including rent division \cite{edward1999rental,gal2017fairest}, item allocation \cite{tadenuma1993fair,aragones1995derivation,klijn2000algorithm,halpern2019fair,BrustleDNSV20,caragiannis2021computing,BarmanKNS22,teh2024envy,aziz2025weighted,la2025discrepancy} and job scheduling \cite{Hartline2008,cohen2010envy,feldman2025proportionally}, with a variety of objectives including bounding payments, measuring social welfare, capping makespan, and approximating minimum subsidy.

\paragraph{Online Fair Division.} The problem of finding fair allocations in online settings has been extensively explored in the literature on online fair division.
Prior work has examined various formulations of the problem, including settings with online items and offline agents \cite{he2019achieving,gkatzelis2021fair,HalpernPVX25}, offline items and online agents \cite{walsh2011online,kash2014no,kulkarni2025online}, as well as randomized mechanisms \cite{AleksandrovAGW15,BenadeKPPZ24}. Primary research objectives include minimizing the maximum envy, bounding the number of online adjustments, approximating social welfare, and approximating the maximin share.

\section{Preliminaries}

We consider a fair division setting with a set of agents $[n] = \{1, 2, \ldots ,n\}$ and a set of items $[m] = \{1, 2, \ldots, m\}$. Each agent $i \in [n]$ is equipped with a valuation function $v_i : 2^{[m]} \to \mathbb{R}_{\geq 0}$. We make the standard assumptions that each valuation is monotone, i.e. $v_i(S) \leq v_i(T)$ for all $S \subseteq T \subseteq [m]$, and normalized, i.e. $v_i(\emptyset) = 0$. Our theorems are written for the setting where the items are goods and the marginal values are therefore nonnegative, but, as we discuss later, a similar analysis can be extended to the case of chores without much further effort.

We consider some important types of instances with additional assumptions on the structure of the valuations, which we define in the corresponding subsections.

Our task is to select an allocation $X = (X_1, X_2, \ldots, X_n)$ and a subsidy vector $p = (p_1, \ldots, p_n)$. An allocation is a partition of $[m]$ into $n$ bundles such that $X_1 \cup X_2 \cup \ldots \cup X_n = [m]$ and $X_i \cap X_j = \emptyset$ for $i \neq j$. A subsidy vector must satisfy $p_i \geq 0$ for all $i \in [n]$.

\subsection{Fairness Notions and Local Efficiency}

Envy-freeness, a central fairness notion in the fair division literature requires that no agent envies the bundle of another agent. An allocation $X = (X_1, \ldots, X_n)$ is envy-free (EF) if $v_i(X_i) \geq v_i(X_j)$ for all $i, j\in[n]$. Since envy-free allocations do not necessarily exist, we need the use of a subsidy.

\begin{definition}[Envy-Freeness with Subsidy, Envy-Freeability]\label{def:ef_subsidy}
    An allocation $X = (X_1, \ldots, X_n)$ together with a subsidy vector $p = (p_1, \ldots, p_n)$ is envy-free if $v_i(X_i) + p_i \geq v_i(X_j) + p_j$ for all $i, j\in[n]$. An allocation $X$ is said to be envy-freeable if a subsidy vector satisfying this condition exists.
\end{definition}

The social welfare of an allocation is the sum of the values obtained by all agents, i.e., $\sw(X) = \sum_{i \in [n]} v_i(X_i)$. A well-known result characterizes envy-freeness via local efficiency.

\begin{definition}[Local Efficiency]\label{def:local_efficiency}
An allocation $X = (X_1, \ldots, X_n)$ is locally efficient if there is no permutation of the agents that increases the social welfare, i.e., $X$ is LE if for all $\pi : [n] \to [n]$ we have $\sum_{i \in [n]} v_i(X_i) \geq \sum_{i \in [n]} v_{\pi(i)}(X_i)$.
\end{definition}

\begin{theorem}[\cite{aragones1995derivation,Haake2002,Hartline2008,halpern2019fair}]\label{thm:ef-le}
An allocation $X = (X_1, \ldots, X_n)$ is envy-freeable if and only if $X$ is locally efficient.
\end{theorem}

Clearly, any allocation that maximizes the social welfare (across all possible allocations) is locally efficient.

\subsection{Online Setting}

We consider an online setting where items $1, 2, \ldots, j$ arrive one by one and must be irrevocably allocated upon arrival. This means that when item $j$ arrives, the items $1, 2, \ldots, j-1$ are already allocated to some agents, and the algorithm must allocate item $j$ before item $j+1$ arrives. We assume that the online algorithm has value-oracle access, meaning that when item $j$ arrives, the algorithm can access any agent's value for any bundle containing the items that have arrived up to that point, i.e., the algorithm can query $v_i(S)$ for any $i \in [n]$ and $S \subseteq \{1, \ldots, j\}$.
We further assume that the algorithm selects a subsidy vector $p = (p_1, \ldots, p_n)$ after all items have been allocated, and we aim to construct algorithms that ensure envy-freeness with subsidy (Definition~\ref{def:ef_subsidy}) at the end of the execution, once all items have arrived.

Note that in our online model, the algorithm does not know the number of items in advance. Any algorithm that guarantees envy-freeness for online instances must maintain local efficiency at every step, i.e., whenever the online arrival of items (unpredictably) terminates, the algorithm must output a locally efficient allocation and a subsidy vector that eliminates all envy. Consequently, rather than updating the subsidies online, the algorithm must only maintain a locally efficient allocation online, along with a minimum envy-eliminating payment vector for that allocation.

\subsection{Minimum Subsidy and the Envy Graph} \label{subsec:sub_min}

For classes of valuations where envy-freeness with subsidies is attainable, our goal is to design online algorithms that use as little total subsidy as possible, formally defined below.\footnote{Note that a simple reduction from the Partition problem shows that deciding whether a given instance has an envy-free allocation is NP-hard. Thus the optimization problem of \textit{minimizing} the subsidy for a given instance is also NP-hard.}

\begin{definition}[Total Subsidy]
A subsidy vector $p = (p_1, \ldots, p_n)$ is said to use total subsidy at most $X$ if $\sum_{i \in [n]} p_i \leq X$.
\end{definition}

Throughout this paper, as in prior work on bounding the subsidy, we assume that the marginal value of any item is at most $1$, i.e., for any $S \subseteq [m]$ and $g \in [m]$, $v_i(S \cup \{g\}) - v_i(S) \leq 1$. This assumption can be made without loss of generality by appropriately scaling all valuations by the same factor, which only changes the relative units in which the subsidy is measured.

For a \textit{given} envy-freeable allocation, it was shown by \cite{halpern2019fair} that an envy-eliminating subsidy vector that minimizes the subsidy for each agent (and therefore also minimizes the total subsidy) can be found in polynomial time. In this work, rather than bounding the per-agent subsidy, we focus on bounding the total subsidy, which is at most $n$ times larger. From the perspective of our notion that a {\em small} subsidy is independent of the number of items, the two objectives are essentially equivalent.

For the analysis of fair allocations with subsidies, we will utilize an auxiliary directed graph defined for every allocation called the \textit{envy graph}. For a given fair division instance, and given allocation $X$, the envy graph $G_X$ is constructed as follows. The agents $[n]$ form the set of vertices of this graph, and there is a weighted directed arc for every ordered pair of agents. The weight on arc $(i,k)$ is the envy that agent $i$ has for agent $k$ in $X$, i.e., $w_{G_X}(i,k) = v_i(X_k) - v_i(X_i)$.

\citet*{halpern2019fair} observed that if an allocation is locally efficient, then the corresponding envy graph does not have any directed cycles of positive weight. Thus the heaviest path starting at any vertex $i\in V(G_X)$ is well-defined.
An important lemma for bounding the total subsidy, proved in \cite{halpern2019fair} for additive valuations and in \cite{BrustleDNSV20} for monotone valuations, is as follows. We provide a proof for completeness.

\begin{lemma}[\cite{halpern2019fair,BrustleDNSV20}]\label{lem:additive-paths}
Given a fair division instance with $m$ goods, $n$ agents, and monotone valuation functions $v_i$, consider an allocation $X=(X_1,\ldots,X_n)$ that is locally efficient. Let $\ell(i)$ denote the weight of the heaviest path in the envy graph $G_X$ starting at vertex $i$. Then the minimum subsidy required to make $X$ envy-free is exactly $\sum_{i \in [n]}\ell(i)$. Moreover, this quantity can be bounded above by $m(n-1)$.
\end{lemma}

\begin{proof}
    Since the allocation $X$ is envy-freeable, from Theorem~\ref{thm:ef-le}, it is locally efficient. Therefore, the envy graph $G_X$ corresponding to this allocation is acyclic. Define $\ell(i)$ as the weight of the heaviest (possibly empty) path starting at agent $i$ in the envy graph. Suppose this path is $(i = i_1, \ldots, i_k)$ for some $k \in [n]$. Then $\ell(i) = \sum_{j \in [k-1]} w_{G_X}(i_j,i_{j+1})$ where $w_{G_X}(i_j,i_{j+1})$ is the weight on the arc from agent $i_j$ to agent $i_{j+1}$, which is equal to the corresponding envy. Now, note that since the maximum marginal value of any good is upper bounded by 1, $w_{G_X}(i_j,i_{j+1}) \leq v_{i_j}(X_{i_{j+1}}) \leq |X_{i_{j+1}}|$. Therefore, $\sum_{j \in [k-1]} w_{G_X}(i_j,i_{j+1}) \leq \sum_{j \in [k-1]}|X_{i_{j+1}}| \leq m$. Further, since the envy graph has no positive weight directed cycles, there is an agent $i^*$ such that $\ell(i^*) = 0$, i.e., all paths starting at this agent have negative weight. Now, if we show that the minimum subsidy is bounded by $\sum_{i \in [n]} \ell(i)$, this implies a bound of $m(n-1)$ on the total subsidy, completing the proof.

    It remains to show that the minimum subsidy is bounded by $\sum_{i \in [n]}\ell(i)$. Suppose we assign to each agent $i$ the subsidy $p_i = \ell(i)$. For any pair $i,j$ of agents, we have $v_i(X_i) + p_i = v_i(X_i) + \ell(i) = (v_i(X_j) - w_{G_X}(i,j)) + \ell(i) \geq v_i(X_j) - w_{G_X}(i,j) + \ell(j) + w_{G_X}(i,j) = v_i(X_j) + \ell(j)$, where the inequality follows from the triangle inequality in the graph $G_X$. This completes the proof.
\end{proof}

\section{Boundaries of Envy-Freeability}

Recall that in order to maintain envy-freeability, we are tasked with maintaining local efficiency online; that is, after each incoming item is assigned to some agent, we require that no permutation of the bundles increases the social welfare of the resulting allocation.

We begin our analysis with additive valuations, which admit an interesting property in the online setting: when a new item arrives, for any locally efficient partial allocation (of the items that have previously arrived), there exists an agent to which the new item can be assigned such that the resulting allocation is locally efficient. In other words, local efficiency can always be extended online.

\begin{proposition} \label{prop:local-eff}
    In a fair division instance with additive valuations, let $S \subseteq [m]$ be a subset of items, and let $g\in [m]\setminus S$ be an item not in $S$. Given any partial allocation $X$ of the items in $S$ that is locally efficient, let $X'$ be the allocation obtained by assigning $g$ to some agent $i^*\in\argmax_{i\in[n]} v_i(\{g\})$. Then $X'$ is locally efficient.
\end{proposition}
\begin{proof}
    Suppose, for a contradiction, that there exists a permutation $\pi : [n] \to [n]$ such that reallocating the bundles of $X'$ along $\pi$ increases the social welfare. Then, since the valuations are additive, we have
    \begin{align}
        \sum_{i \in [n]} v_{\pi(i)}(X'_i) &> \sum_{i \in [n]} v_i(X'_i)\nonumber\\
        \Longrightarrow \quad v_{\pi(i^*)}(X'_{i^*}\setminus\{g\}) &+ \sum_{i \in [n]\setminus \{i^*\}} v_{\pi(i)}(X'_i) \nonumber\\
        >\;v_{i^*}(X'_{i^*}\setminus\{g\}) &+ \sum_{i \in [n]\setminus \{i^*\}} v_i(X'_i) \label{eq1}\\
        \Longrightarrow \quad \sum_{i \in [n]} v_{\pi(i)}(X_i) &> \sum_{i \in [n]} v_i(X_i),\nonumber
    \end{align}
    which contradicts the fact that $X$ is locally efficient.
\end{proof}

The above proposition immediately gives us a simple algorithm for maintaining local efficiency online when the agents have additive valuations: assign each item to an agent that has the highest marginal value for it, thereby maximizing the social welfare. Observe that if the valuations are not additive, the implication in (\ref{eq1}) is no longer necessarily true. This raises the natural question of whether local efficiency can be maintained online for classes beyond additive valuations. We answer this question either affirmatively or negatively for a variety of valuation classes. We begin with our positive results.

\subsection{LE Can Be Maintained Online}

We show that LE can be maintained online for the class of SPLC valuations, and the class of $k$-demand valuations for any $k$ (i.e., we allow $k$ to be as large as $m$). These valuation classes each generalize the class of additive valuations.

\subsubsection{SPLC Valuations}\label{sec:splc}

A valuation function $v_i$ is SPLC (Separable Piecewise-Linear Concave), if the set of goods can be partitioned into $T$ types such that any two items of the same type are identical, for type $t\in[T]$ there is a set $M_t$ of goods of that type, and the agent has a value $v^t_{i,\ell}$ associated with the $\ell^\text{th}$ copy of any good of type $t$. Across different types, the valuation function is additive. The function is concave, so that $v^t_{i,1} \geq v^t_{i,2} \geq \cdots \geq v^t_{i,|M_t|}$. Putting this all together, the value for a given bundle $S$ is $v_i(S) = \sum_{t\in[T]} \sum_{\ell\in [|S\cap M_t|]} v^t_{i,\ell}$. These valuations, also called additively separable concave valuations, have been previously studied in the context of fair division \cite[e.g.][]{anari2018nash,chaudhury2022fair,chekuri20241}.

\begin{theorem}\label{thm:splc}
    Envy-freeability can be maintained online for all instances with SPLC valuations.
\end{theorem}
\begin{proof}
    The algorithm for maintaining LE in this case is simple: it assigns each incoming good to the agent with the highest marginal value for it. We show that, even for SPLC valuations, this algorithm finds a social welfare maximizing allocation and is therefore locally efficient. The social welfare of an allocation $X = (X_1, \ldots, X_n)$ is
    \begin{align*}
        \sw(X) &= \sum_{i \in [n]} v_i(X_i) \\
        &= \sum_{i \in [n]} \sum_{t\in[T]} \sum_{\ell\in [|S\cap M_t|]} v^t_{i,\ell} \\
        &= \sum_{t\in[T]} \sum_{i \in [n]} \sum_{\ell\in [|S\cap M_t|]} v^t_{i,\ell}
    \end{align*}
    where we swap the order of summation in the last equality. To prove that our allocation maximizes welfare, it therefore suffices to prove that for each type $t \in [T]$, it maximizes $\sum_{i \in [n]} \sum_{\ell\in [|S\cap M_t|]} v^t_{i,\ell}$.
    
    Now, fix a type $t \in [T]$. This type has a total of $|M_t|$ copies in the instance, so each agent has a set of $|M_t|$ marginal values for different copies of this good for a total of $n|M_t|$ possible marginal values. These valuations are concave i.e., each additional copy has a decreasing marginal value for any agent. Since the algorithm allocates each incoming copy to the agent with the highest marginal value, the algorithm will choose the highest $|M_t|$ marginal values for allocating this good, which maximizes the social welfare amongst all possible assignments.
\end{proof}

\subsubsection{$k$-Demand Valuations}
\label{sec:kdemand}
A valuation function $v_i$ belongs to the class of $k$-demand valuations if the value of any bundle $S$ is equal to the sum of the values of the $k$ highest-valued items, valued as singletons, in $S$. That is, every item $j\in[m]$ has an associated value $v_i(\{j\})$, and the value of a non-empty bundle $S$ is $v_i(S) = \max_{T\subseteq S:|T|\leq k} \sum_{j\in T}v_i(\{j\})$. This class notably generalizes the well-studied class of unit-demand valuations (when $k=1$), and the class of additive valuations (when $k\geq m$), and models scenarios where an agent can hold at most $k$ items; see \cite[e.g.][]{ByrkaMP25,ZhangC20,DeligkasMS21}.

Interestingly, unlike the case of additive valuations or SPLC valuations, for $k$-demand valuations it is impossible to maintain a social-welfare-maximizing allocation online, as shown in the following example. Consider an instance with unit-demand valuations, with $n=2$ agents and $m=2$ items, where agent 1 has values 0.75 and 1 respectively for the two items, and agent 2 has values 0.5 and 0.25 respectively for the two items. The welfare-maximizing allocation for this instance gives the second item to agent 1 and the first item to agent 2, for a total value of 1.5. However, in the online setting, when item 1 arrives, it is assigned to agent 1 (since $0.75 > 0.5$) to maximize welfare. Nonetheless, we show that for $k$-demand valuations it is possible to maintain local efficiency (and envy-freeability) in the online setting.

\begin{theorem}
\label{thm:le-k-demand}
    Envy-freeability can be maintained online for all instances with $k$-demand valuations.
\end{theorem}
\begin{proof}
Consider the following algorithm. Each time a good $g$ arrives, the algorithm allocates it to the agent $i_g$ with the highest singleton value for it, i.e., $i_g \in \arg\max_{i \in [n]} v_i(\{g\})$. Let the resulting allocation be $X = (X_1, \ldots, X_n)$. When the agents have $k$-demand valuations, at every step in the execution of the online algorithm, this allocation ensures that each bundle is assigned to the agent with the highest value for this bundle, i.e., $v_i(X_i) \ge v_{i'}(X_i)$ for all $i, i'\in [n]$. To see this, fix any pair of agents $i, i' \in [n]$. By the definition of $k$-demand valuations, $v_{i'}(X_i) = \max_{T \subseteq X_i: |T| \le k} \sum_{j \in T} v_{i'}(\{j\})$. 
Let $T^\star$ be a subset of $X_i$ that achieves this maximum. Then $v_{i'}(X_i) = \sum_{j \in T^\star} v_{i'}(\{j\})$, which by our algorithm is at most $\sum_{j \in T^\star} v_i(\{j\}) \le v_i(X_i)$, where the last inequality follows from the definition of $k$-demand valuations. Thus $v_i(X_i) \ge v_{i'}(X_i)$ for all $i, i' \in [n]$, so there is no reassignment of the $n$ fixed bundles across agents that increases the social welfare. Consequently, the algorithm maintains local efficiency.
\end{proof}

\subsection{LE Cannot Be Maintained Online}

In this section, we show that it is impossible to maintain local efficiency online for budget-additive valuations, even if the budgets are known in advance, and for submodular and supermodular valuations, even if the marginal values are binary (i.e., in $\{0,1\}$). The impossibilities for binary submodular and binary supermodular valuations are particularly surprising, given the large collection of positive results in fair division for these highly structured valuation classes \cite[e.g.][]{barman2020existence,babaioff2021fair,barman2023fair}.

\subsubsection{Budget-Additive Valuations} 

A valuation function $v_i$ is budget-additive if it is additive, but capped at an agent-specific budget $B_i$. That is, the value of a non-empty bundle $S$ is $v_i(S) = \min\{\sum_{j\in S}v_i(\{j\}),B_i\}$. Budget-additive valuations have been widely studied in auctions and fair division \cite[e.g.][]{mehta2007adwords,buchfuhrer2010inapproximability,garg2018approximating}.

\begin{theorem}\label{thm:budget_additive}
No online algorithm can maintain envy-freeability in instances with budget-additive valuations, even if the budgets are known in advance.
\end{theorem}
\begin{proof}
Consider an instance with two agents where agent $1$'s budget is $1-\epsilon$, and agent $2$'s budget is $1$. Suppose an item arrives online with value $1-\epsilon$ for agent $1$ and value $1-2\epsilon$ for agent $2$. Any online algorithm that maintains LE assigns this item to agent $1$. Next, suppose a second item arrives online with value $1-\epsilon$ for agent $1$ and value $1/2$ for agent $2$. If this item is assigned to agent $1$, the resulting social welfare is $1-\epsilon$, but swapping the two bundles results in an increase in the social welfare to $1$. Similarly, if the second item is instead assigned to agent $2$, the resulting social welfare is $3/2 - \epsilon$, but swapping the two bundles results in a social welfare of $2-3\epsilon$, which is an increase for sufficiently small $\epsilon>0$.
\end{proof}

    \subsubsection{Binary Submodular Valuations} Next, we show that local efficiency cannot be maintained online for the class of submodular valuations. This result holds even in the special case where the marginal value of an item is either 0 or 1, i.e., submodular valuations with binary marginals, and therefore extends to matroid-rank valuations and gross-substitutes valuations \cite{Leme17}. It is well-known that the class of submodular valuations with binary marginals exactly coincides with the class of matroid-rank valuations.

A valuation function $v_i$ is matroid rank if there is an underlying finite matroid $\mathcal{M}_i = ([m], \mathcal{I}_i)$ such that the value of a bundle $S \subseteq [m]$ is equal to the rank of that bundle in $\mathcal{M}_i$, i.e., the size of the largest independent subset of $S$. In the following proof, we will use the fact that matroids are closed under restriction, i.e., if an online instance with a given profile of matroid-rank valuations terminates early, then the resulting partial instance is a feasible instance for the class of matroid-rank valuations.

\begin{theorem}\label{thm:binary_submodular}
No online algorithm can maintain envy-freeability in instances with binary submodular valuations.
\end{theorem}
\begin{proof}
For a given online algorithm, we will adversarially construct an instance with matroid-rank valuations for which this algorithm cannot maintain local efficiency online. The instance has two agents, labeled 1 and 2, and up to four items, labeled $a, b, c$ and $d$, that arrive in order. Suppose both agents have a value of 1 for each of the items $a$ and $b$ and value 2 for the bundle $\{a,b\}$. Without loss of generality, suppose the algorithm assigns item $a$ to agent 1.

When item $b$ arrives, we have the following cases.

\underline{Case 1:}
\textit{The algorithm assigns item $b$ to agent $2$.}
Suppose item $c$ arrives next, and agent $1$'s underlying matroid has the bases $\{\{a,b\},\{b,c\}\}$, and agent $2$'s matroid has the bases $\{\{a,b\},\{a,c\}\}$ (this choice is consistent with the valuations of $a$ and $b$). If item $c$ is assigned to agent $1$, then the social welfare of the resulting allocation is 2. However, swapping the bundles increases the social welfare to 3. Similarly, if item $c$ is assigned to agent $2$, then swapping the bundles increases the social welfare from 2 to 3. Consequently, the online algorithm fails to maintain local efficiency.

 \underline{Case 2:}
 \textit{The algorithm assigns item $b$ to agent $1$.} We will now require item $d$. Suppose agent $1$'s underlying matroid has the bases $\{\{a,b\},\{a,c\},\{a,d\},\{b,c\},\{b,d\},\{c,d\}\}$, and agent $2$'s underlying matroid has the bases $\{\{a,b,c\},\{a,b,d\}\}$ (these choices are consistent with  valuations of $a$ and $b$). We have two cases, corresponding to the possible assignments of item $c$.

 \underline{Case 2a:}
 \textit{The algorithm assigns item $c$ to agent $1$.} If item $c$ is assigned to agent $1$, then the social welfare of the resulting allocation is 2. However, swapping the bundles increases the social welfare to 3, so the online algorithm fails to maintain local efficiency in the restricted instance.

\underline{Case 2b:}
\textit{The algorithm assigns item $c$ to agent $2$.} 
Suppose item $c$ is assigned to agent $2$. Then, when item $d$ appears, it is assigned to either agent $1$ or agent $2$. If $d$ is assigned to agent $1$, then the social welfare of the resulting allocation is 3. However, swapping the bundles increases the social welfare to 4, so the online algorithm fails to maintain local efficiency. For the other case, suppose $d$ is assigned to agent $2$, then the social welfare of the resulting allocation is again 3. However, swapping the bundles once again increases the social welfare to 4, so the online algorithm fails to maintain local efficiency in all cases.
\end{proof}

\subsubsection{Binary Supermodular Valuations.}

We present a similar result for binary supermodular valuations. However, the argument additionally requires the use of a fifth item $e$.

\begin{theorem}\label{thm:binary_supermodular}
No online algorithm can maintain envy-freeability in instances with binary supermodular valuations.
\end{theorem}

\begin{proof}
For a given online algorithm, we will adversarially construct an instance with binary supermodular valuations for which this algorithm cannot maintain local efficiency online. The instance has two agents, labeled 1 and 2, and up to five items, labeled $a, b, c$, $d$ and $e$, that arrive in order. Suppose both agents have a value of 0 for each of the items $a$ and $b$ and for the bundle $\{a,b\}$. Without loss of generality, suppose the algorithm assigns item $a$ to agent 1.

When item $b$ arrives, we have the following cases.

\underline{Case 1:} 
\textit{The algorithm assigns item $b$ to agent $2$.} 
Suppose item $c$ arrives next, agent $1$'s valuation function assigns a value of 1 for the sets $\{\{b,c\},\{a,b,c\}\}$ and a value of 0 for the other sets, and agent $2$'s valuation function assigns a value of 1 for the sets $\{\{a,c\},\{a,b,c\}\}$ and a value of 0 for the other sets. It is easy to verify that all marginal values are either 0 or 1, and that these functions are supermodular. If item $c$ is assigned to agent $1$, the resulting allocation has a social welfare of 0; however, swapping the bundles increases the welfare to 1. Similarly, if $c$ is assigned to agent $2$, swapping the bundles increases the welfare from 0 to 1.

\underline{Case 2:} 
\textit{The algorithm assigns item $b$ to agent $1$.} 
We will now require items $d$ and $e$ to complete the argument. Suppose item $c$ arrives and all subsets of items have a value of 0 for both agents. We have two cases, corresponding to the possible assignments of item $c$.

\underline{Case 2a:}
\textit{The algorithm assigns item $c$ to agent $2$.} 
Suppose the algorithm assigns $c$ to agent 2. When item $d$ arrives, suppose agent 1 has a value of 1 for the sets $\{\{c,d\},\{a,c,d\},\{b,c,d\},\{a,b,c,d\}\}$, and agent 2 has a value of 1 for the sets $\{\{a,b,d\},\{a,b,c,d\}\}$. It can be verified that the resulting functions are binary supermodular. If $d$ is assigned to agent 1, swapping the bundles increases the social welfare from 0 to 1. Similarly, if $d$ is assigned to agent 2, swapping the bundles increases the social welfare from 0 to 1. In either case, the online algorithm fails to maintain local efficiency.

\underline{Case 2b:}
\textit{The algorithm assigns item $c$ to agent $1$.} 
Suppose item $c$ is assigned to agent $1$. Then, when item $d$ appears, suppose agent 1 has a value of 0 for all subsets, and agent 2 has a value of 1 for only the set $\{a,b,c,d\}$ and a value of 0 otherwise. If item $d$ is assigned to agent 1, a bundle swap increases the social welfare from 0 to 1, so $d$ is assigned to agent 2. Then, suppose item $e$ arrives, and suppose agent 1 has a value of 1 for the set $\{d,e\}$ and all of its supersets (and a value of 0 otherwise), and agent 2 has a value of 1 for the sets $\{\{a,b,c,d\},\{a,b,c,e\}\}$ and a value of 2 for the set $\{a,b,c,d,e\}$. It can be verified that the resulting valuations are binary supermodular. Now, if $e$ is assigned to agent 1, swapping the bundles increases the welfare from 0 to 1. If instead $e$ is assigned to agent 2, swapping the bundles again increases the welfare from 0 to 1. The online algorithm again fails to maintain local efficiency in all cases, completing the argument.
\end{proof}

We remark that it is much easier to construct small counterexamples for the classes of submodular and supermodular valuations if the assumption of binary marginal values is dropped.

\section{(Mostly) Tight Subsidy Bounds}

In the previous section, our goal was to delineate the valuation classes for which it is possible to maintain local efficiency online. For this section, our objective is to maintain local efficiency while additionally bounding the total subsidy sufficient for envy-freeness in the resulting allocation. We investigate a variety of valuation classes in order to determine tight upper bounds and lower bounds on the subsidy for these classes.

As a consequence of Lemma~\ref{lem:additive-paths} the upper bound of $m(n-1)$ holds for all valuation functions that we consider in this section. First, in Section~\ref{sec:high-bounds}, we show that for additive valuations (and therefore its superclasses, including SPLC valuations and $k$-demand valuations for unrestricted $k$), this is the best subsidy bound that any online algorithm can achieve. Next, in Section~\ref{sec:low-bounds} we investigate structured classes of valuations and show that this subsidy bound improves. In particular, we show that a {\em small} subsidy, independent of the number of items $m$, is possible for these classes. Most of our subsidy bounds in this section are exactly tight even up to coefficients.

\subsection{Large Subsidy Bounds}\label{sec:high-bounds}

In this section, we consider the case of additive and, more generally, SPLC valuations, as discussed in Section~\ref{sec:splc}.

\begin{theorem}\label{thm:additive-lb}
There exists an online algorithm that achieves envy-freeness with a total subsidy of at most $m(n-1)$ for all instances with SPLC valuations, which in particular include all additive valuation instances. Moreover, for any $\epsilon > 0$ and any online algorithm, there exists an additive instance for which achieving envy-freeness requires a total subsidy of at least $m(n-1) - \epsilon$.
\end{theorem}

\begin{proof} 
\begin{table}[t]
    \centering
    \begin{tabular}{|c|| c | c | c | c | c |}
         \hhline{|-||-----|}
             & 
        \textbf{Item} $\bm{1}$ & \textbf{Item} $\bm{2}$ & $\bm{\cdots}$ & \textbf{Item} $\bm{m-1}$ & \textbf{Item} $\bm{m}$\\
         \hhline{|-||-----|}
         \noalign{\vskip 2pt}
        \hhline{|-||-----|}
        Agent 1 & \cellcolor{orange!20}$1-\Bar{\epsilon}+2\delta$  & \cellcolor{orange!20}$1-\Bar{\epsilon}+4\delta$ & \cellcolor{orange!20}$\cdots$ & \cellcolor{orange!20}$1-\Bar{\epsilon}+{2^{m-1}}\delta$ & \cellcolor{orange!20}$1-\Bar{\epsilon}+{2^{m}}\delta$\\
                 \hhline{|-||-----|}  \noalign{\vskip 2pt}
        \hhline{|-||-----|}
        Agent 2 & \cellcolor{gray!20}$1-\Bar{\epsilon}+\delta$ & \cellcolor{gray!20}$1-\Bar{\epsilon}+2\delta$ & \cellcolor{gray!20}$\cdots$ & \cellcolor{gray!20}$1-\Bar{\epsilon}+{2^{m-2}}\delta$ & \cellcolor{gray!20}$1-\Bar{\epsilon}+{2^{m-1}}\delta$\\
                 \hhline{|-||-----|}
         Agent 3 & \cellcolor{gray!20}$1-\Bar{\epsilon}+\delta$ & \cellcolor{gray!20}$1-\Bar{\epsilon}+2\delta$ & \cellcolor{gray!20}$\cdots$ & \cellcolor{gray!20}$1-\Bar{\epsilon}+{2^{m-2}}\delta$ & \cellcolor{gray!20}$1-\Bar{\epsilon}+{2^{m-1}}\delta$\\
                  \hhline{|-||-----|}
$\raisebox{0.75ex}{$\vdots$}$ & \cellcolor{gray!20}$\raisebox{0.75ex}{$\vdots$}$ & \cellcolor{gray!20}$\raisebox{0.75ex}{$\vdots$}$ & \cellcolor{gray!20}$\raisebox{0.75ex}{$\ddots$}$ & \cellcolor{gray!20}$\raisebox{0.75ex}{$\vdots$}$ & \cellcolor{gray!20}$\raisebox{0.75ex}{$\vdots$}$\\
                  \hhline{|-||-----|}
         Agent $n$ &\cellcolor{gray!20} $1-\Bar{\epsilon}+\delta$ & \cellcolor{gray!20} $1-\Bar{\epsilon}+2\delta$ & \cellcolor{gray!20}$\cdots$ &\cellcolor{gray!20} $1-\Bar{\epsilon}+{2^{m-2}}\delta$ & \cellcolor{gray!20}$1-\Bar{\epsilon}+{2^{m-1}}\delta$\\
         \hhline{|-||-----|}
    \end{tabular}
    \caption{Valuation function showing the lower bound for general additive valuation functions. 
   The value in the $i^{th}$ row and $j^{th}$ column is the value of agent $i$ for good $j$.
The orange agent-item pairs correspond to the unique allocation maintaining envy-freeness at every step, and the sum of gray entries is the required subsidy to make it envy-free.
    }
    \label{tab:add-lower-bound-val}
\end{table}

\emph{Upper bound.} This part follows directly from Lemma~\ref{lem:additive-paths}.

\emph{Lower bound.} 
Consider an instance with $n$ agents and $m$ goods. The valuations of the agents are as shown in Table~\ref{tab:add-lower-bound-val}. The rows correspond to the agents and the columns correspond to goods. The goods arrive in order $1$ to $m$. The value of the $j^{th}$ good for agent $1$ is $1-\Bar{\epsilon} + 2^{j}\delta$, and for each agent $2 \leq i \leq n$ it is $1-\Bar{\epsilon} + 2^{j-1}\delta$. We choose $\delta$ and $\Bar{\epsilon}$ such that $\Bar{\epsilon} = 2^{m}\delta$, so that the maximum value of an item is 1. Note that each incoming good has a higher value for all agents than any previous good, and its value is always largest for agent $1$.

We claim that any allocation that maintains local efficiency allocates all goods to agent $1$. For the base case, the first good is allocated to agent $1$, since this agent has the highest value for it. Now suppose that the first $j$ goods have all been allocated to agent $1$. When the $(j+1)^{\text{st}}$ good arrives, suppose we assign this good to some other agent, say agent $2$. Then agent $1$ receives a total value of $j(1-\Bar{\epsilon}) + \sum_{r \in [j]}2^{r}\delta = j(1-\Bar{\epsilon}) + (2^{j+1} - 2)\delta$, while agent $2$ receives a value of $(1-\Bar{\epsilon}) + 2^{j}\delta$, and all other agents receive no items. The total value of this allocation is therefore $(j+1)(1-\Bar{\epsilon}) + (2^{j+1} + 2^{j} - 2)\delta$.

If we swap the bundles between agent $1$ and agent $2$, agent $1$ receives $(1-\Bar{\epsilon}) + 2^{j+1}\delta$ and agent $2$ receives $j(1-\Bar{\epsilon}) + (2^{j} - 1)\delta$, giving a total value of $(j+1)(1-\Bar{\epsilon}) + (2^{j+1} + 2^{j} - 1)\delta$, which is an increase of $\delta$ over the total value before this swap. Consequently, the previous allocation is not locally efficient, and good $j+1$ is assigned to agent $1$ in order to maintain local efficiency.

Thus, on termination, each agent $i \in \{2,\ldots,n\}$ receives at least $(n-1)m(1-\Bar{\epsilon})$ subsidy, which yields the desired bound by setting $\Bar{\epsilon} = \epsilon / (n(m-1))$.
\end{proof}

\subsection{Small Subsidy Bounds}\label{sec:low-bounds}
In this section, we consider five well-studied valuation classes for which we establish small subsidy bounds. These upper and lower bounds are exactly tight, except for the case of $k$-valued valuations.

\subsubsection{$k$-Demand Valuations}
We begin by analyzing the case of $k$-demand valuations, as discussed in Section~\ref{sec:kdemand}.

\begin{theorem}\label{thm:k-demand-subsidy}
There exists an online algorithm that achieves envy-freeness with a total subsidy of at most $k(n-1)$ for all instances with $k$-demand valuations. Moreover, for any $\epsilon > 0$ and any online algorithm, there exists a $k$-demand instance for which achieving envy-freeness requires a total subsidy of at least $k(n-1) - \epsilon$.
\end{theorem}

\begin{proof}
    \textit{Upper Bound.} We reuse the algorithm from the proof of Theorem~\ref{thm:le-k-demand}, i.e., we allocate each incoming good to an agent that has the largest singleton value for it. Let $X$ be the resulting partial allocation at any step of this algorithm's execution. To see the subsidy bound, we first fix an agent $i \in [n]$ and consider any path in the envy graph starting at $i$. Denote the agents on this path by $i = i_1, \ldots, i_r = j$. Then the weight on the arc $\{i_\ell,i_{\ell+1}\}$ is $w_{G_X}(i_\ell,i_{\ell+1})$. Since any agent's value for its own bundle is at least as much as its value for any other agent's bundle, we get $v_j(X_j) \geq v_{j-1}(X_{j-1}) + w(j-1,j) \geq v_{j-2}(X_{j-2}) + w(j-2,j-1)+w(j-1,j) \geq \cdots \geq v_i(X_i) + \sum_{l \in [r-1]}w(l,l+1) \geq \sum_{\ell \in [r]}w_{G_X}(i_\ell,i_{\ell+1})$. By the definition of $k$-demand valuations, we know $v_i(X_i) \leq k$. Therefore $\sum_{\ell \in [r-1]}w_{G_X}(i_\ell,i_{\ell+1}) \leq k$, and by Lemma~\ref{lem:additive-paths} the total subsidy is bounded by the sum of heaviest paths starting at any node, which is at most $k(n-1)$.

    \textit{Lower Bound.} By Theorem~\ref{thm:additive-lb}, any online algorithm for additive instances with $n$ agents and $k$ items requires a subsidy of at least $k(n-1) - \epsilon$ for any $\epsilon > 0$. Since $k$-demand valuations coincide with additive valuations when $k=m$, the same lower bound applies to $k$-demand valuations.
\end{proof}

When $k$ does not depend on $m$, i.e., for small values of $k$, this upper bound on the subsidy is small. Notably, this allocation rule does not necessarily maximize welfare for this class of valuations, yet as we show, it preserves local efficiency.

\subsubsection{Personalized $k$-Valued Valuations} 
A valuation function $v_i$ is \emph{$k$-valued} if there exist $a_1, \ldots, a_k$ such that $v_i$ is additive and $v_i(\{g\}) \in \{a_1, \ldots, a_k\}$ \cite[e.g.][]{GargM23}. An important special case is the \emph{bivalued} case, corresponding to $k = 2$ \cite[e.g.][]{AmanatidisBFHV21}. Our results apply to the setting of \textit{personalized} $k$-valued valuations, where the set of values $\{a_1, \ldots, a_k\}$ may differ across agents.

\begin{theorem}\label{thm:k-valued}
There exists an online algorithm that achieves envy-freeness with a total subsidy of at most $n^2k^n$ for all instances with personalized $k$-valued valuations.
\end{theorem}
\begin{proof}
 In personalized $k$-valued valuations, each good can take at most $k$ different values for each agent. Therefore, if we represent each good as an $n$-dimensional vector where the $i^{th}$ coordinate denotes the value that agent $i$ assigns to the good, the set of all possible goods can be represented by at most $k^n$ distinct vectors. We refer to each of these vectors as a \emph{type}.
 
The algorithm proceeds as follows. For every good, we map it to its corresponding type. The first good of a given type is allocated to the agent who values it the most, the second to the agent who values it second-highest, and so on, in a round-robin fashion, breaking ties arbitrarily.

We first show that the resulting allocation is locally efficient. Suppose, for a contradiction, that the envy graph has a positive-weight directed cycle (corresponding to a violation of local efficiency). Then, since the allocation can be additively decomposed into $k^n$ separate allocations, one corresponding to each type, there exists an item type such that the allocation of the items of that type violates local efficiency. However, within each type, the items are assigned in non-increasing order of agent values, thus any agent that has a greater value for some item type receives a (weakly) larger number of items of that type, so by the rearrangement inequality~\cite{wikipedia:rearrangement_inequality} the allocation within each type is locally efficient.
 
It remains to bound the subsidy for the above algorithm. For any given type, the envy between any two agents is bounded by the value of a single good. Consequently, the total envy between any two agents across all types is bounded by the total number of types, i.e., $k^n$. Thus the weight of the heaviest path starting at any vertex is at most $(n-1)k^n$, and by Lemma~\ref{lem:additive-paths}, the total subsidy is bounded by $n^2k^n$.
\end{proof}

\subsubsection{Rank-One Valuations}
An instance with valuation functions $v_1, \ldots, v_n$ is \emph{rank-one} if $v_i(X) = \sum_{j \in X} w_i \cdot q_j$ for some agent weights $w_1, \ldots, w_n \in [0,1]$ and base values $q_1, \ldots, q_m \in [0,1]$.
An important motivation for rank-one valuations arises in the sponsored search model, where agents correspond to advertisers and items to ad slots. Here, an advertiser’s value for a slot is the product of their value per click and the slot’s click-through rate \cite{LahaiePSV07}.

\begin{theorem}\label{thm:single_param}
There exists an online algorithm that achieves envy-freeness with a total subsidy of at most $n(n+1)/2-1$ for all instances with rank-one valuations. Moreover, for any $\epsilon > 0$ and any online algorithm, there exists a rank-one instance for which achieving envy-freeness requires a total subsidy of at least $n(n+1)/2-1-\epsilon$.
\end{theorem}
\begin{proof}
    \begin{figure}[t]
    \centering
    \begin{tikzpicture}
    [x=2cm,y=1.75cm]
      \def\barh{0.3}    
      \def\gap{0.05}      
      \def\xincr{0.03}    
      \def\xgap{0.05}    
      \def\basewidth{0.75}  
      \def\red!50shorten{0.05}

      \newcommand{\BarWidth}[1]{\basewidth + (#1 - 1)*\xincr}

      \newcommand{\DrawBar}[3]{%
        \pgfmathsetmacro{\w}{\basewidth + (#3 - 1)*\xincr}%
        \pgfmathsetmacro{\cx}{#1 + 0.5*\w}%
        \pgfmathsetmacro{\cy}{#2 + 0.5*\barh}%
        \pgfmathsetmacro{\shade}{10}%
        \draw[rounded corners=1pt, fill=blue!\shade] (#1,#2) rectangle ++(\w,\barh);%
        \node at (\cx,\cy) {#3};%
        \pgfmathsetmacro{#1}{#1 + \w + \xgap}%
      }

      \pgfmathsetmacro{\yOne}{0}
      \pgfmathsetmacro{\yTwo}{\yOne - (\barh+\gap)}
      \pgfmathsetmacro{\yThree}{\yTwo - (\barh+\gap)}
      \pgfmathsetmacro{\yFour}{\yThree - (\barh+\gap)}
      \pgfmathsetmacro{\yFive}{\yFour - (\barh+\gap)}

      \pgfmathsetmacro{\x}{0}
      \DrawBar{\x}{\yOne}{1}
      \DrawBar{\x}{\yOne}{2}
      \DrawBar{\x}{\yOne}{4}
      \DrawBar{\x}{\yOne}{7}
      \DrawBar{\x}{\yOne}{11}
      \pgfmathsetmacro{\xSevenEnd}{\x}
      \pgfmathsetmacro{\xSevenEndShort}{\xSevenEnd - \red!50shorten}

      \pgfmathsetmacro{\x}{0}
      \DrawBar{\x}{\yTwo}{3}
      \DrawBar{\x}{\yTwo}{5}
      \DrawBar{\x}{\yTwo}{8}
      \pgfmathsetmacro{\xend}{\x}
      \draw[rounded corners=1pt, fill=red!20] (\xend,\yTwo) rectangle ++(\xSevenEndShort-\xend,\barh);

      \pgfmathsetmacro{\x}{0}
      \DrawBar{\x}{\yThree}{6}
      \DrawBar{\x}{\yThree}{9}
      \pgfmathsetmacro{\xend}{\x}
      \draw[rounded corners=1pt, fill=red!20] (\xend,\yThree) rectangle ++(\xSevenEndShort-\xend,\barh);

      \pgfmathsetmacro{\x}{0}
      \DrawBar{\x}{\yFour}{10}
      \pgfmathsetmacro{\xend}{\x}
      \draw[rounded corners=1pt, fill=red!20] (\xend,\yFour) rectangle ++(\xSevenEndShort-\xend,\barh);

      \draw[rounded corners=1pt, fill=red!20] (0,\yFive) rectangle ++(\xSevenEndShort,\barh);

      \pgfmathsetmacro{\xlab}{0.2}

      \pgfmathsetmacro{\yc}{\yOne + 0.5*\barh}
      \node[anchor=mid east,inner sep=0,outer sep=0] at (-\xlab,\yc) {Agent 1};

      \pgfmathsetmacro{\yc}{\yTwo + 0.5*\barh}
      \node[anchor=mid east,inner sep=0,outer sep=0] at (-\xlab,\yc) {Agent 2};

      \pgfmathsetmacro{\yc}{\yThree + 0.5*\barh}
      \node[anchor=mid east,inner sep=0,outer sep=0] at (-\xlab,\yc) {Agent 3};

      \pgfmathsetmacro{\yc}{\yFour + 0.5*\barh}
      \node[anchor=mid east,inner sep=0,outer sep=0] at (-\xlab,\yc) {Agent 4};

      \pgfmathsetmacro{\yc}{\yFive + 0.5*\barh}
      \node[anchor=mid east,inner sep=0,outer sep=0] at (-\xlab,\yc) {Agent 5};

    \end{tikzpicture}%
    \caption{Illustration of the lower bound construction for rank-one valuations in the proof of Theorem~\ref{thm:single_param}. Each blue bar represents an item, with its length proportional to the item’s base value and the label inside indicating its arrival order. The allocation shown is the subsidy-minimizing one that maintains envy-freeability at every step. The red bar lengths are proportional to each agent's envy toward agent $1$.}
    \label{fig:agents-row}
\end{figure}
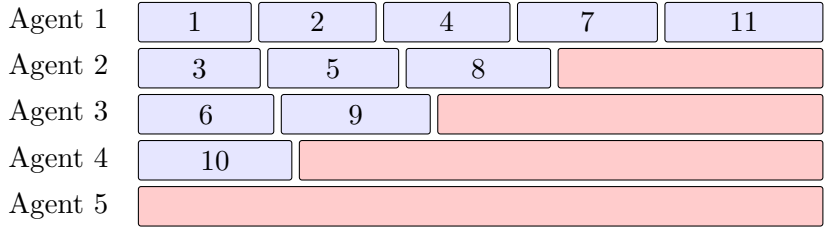

Throughout the proof, we denote $q(S) = \sum_{g \in S} q_g$.

\emph{Upper bound.} 
Let $X$ denote the current allocation and $g$ be the arriving item. The algorithm maintains the invariant $q(X_1) \geq q(X_2) \geq \ldots \geq q(X_n)$, which implies local efficiency by the rearrangement inequality~\cite{wikipedia:rearrangement_inequality}.
The allocation rule for each arriving item $g$ is defined as follows. If there exists an agent $i \in [n-1]$ such that $q(X_i) \geq q(X_{i+1}) + 1$, then the algorithm allocates $g$ to agent $i+1$. If no such agent exists, the algorithm allocates $g$ to agent $1$.

We first observe that, by construction, at any point in the execution of the algorithm there can be at most one index $i$ such that $q(X_i) \geq q(X_{i+1}) + 1$. Moreover, for agent $i$, $q(X_i) \leq q(X_{i+1}) + 2$.

By the characterization of the minimum subsidy (Lemma~\ref{lem:additive-paths}), the subsidy required to make the allocation envy-free is equal to the weight of the heaviest path starting at $i$ in the envy graph. The weight of a path $i_1 \to i_2 \to \cdots \to i_k$ is given by $\sum_{r=1}^{k-1} w_r \cdot (q(X_{r+1}) - q(X_r))$.

We now argue that, without loss of generality, the heaviest path can be assumed to have the form $i \to i-1 \to \cdots \to 1$. Indeed, consider any arc $i \to j$ with $i > j$. This edge can be replaced by the path $i \to i-1 \to \cdots \to j+1 \to j$, whose weight is $\sum_{r=j+1}^i w_r \cdot (q(X_{r-1}) - q(X_r)) \geq \sum_{r=j+1}^i w_i \cdot (q(X_{r-1}) - q(X_r)) = w_i \cdot (q(X_i) - q(X_j))$. Thus, this replacement cannot decrease the path weight. Furthermore, since the envy graph contains no positive-weight cycle by Lemma~\ref{lem:additive-paths}, the heaviest path is simple. Combined with the invariant $q(X_1) \geq q(X_2) \geq \cdots \geq q(X_n)$, this implies that the heaviest path indeed has the structure $i \to i-1 \to \cdots \to 1$.

For such a path, the subsidy required for agent $i$ is at most
 $\sum_{r=2}^i w_r \cdot (q(X_{r-1}) - q(X_r)) \leq \sum_{r=2}^i (q(X_{r-1}) - q(X_r)) = q(X_1) - q(X_i) \leq i$ where the final inequality follows from the observation above.
Thus each agent $i$ with $2 \leq i \leq n$ receives a subsidy of at most $i$, and the total subsidy is bounded by $2 + 3 + \ldots + n = n(n+1)/2 - 1$.

\emph{Lower bound.} Fix any online algorithm for the class of rank-one valuations.
Consider the following rank-one instance with $n$ agents and $m = n(n+1)/2$ goods. The agents have parameters $w_i = 1 - i \cdot \epsilon$ for $i \in [n]$, and the goods have parameters $q_j = 1 - \epsilon + 2^{j - n} \cdot \epsilon$ for $j \in [m]$, where $0 < \epsilon < 1/n$.

We first claim that for any $i \in [n-1]$, throughout the execution of the algorithm, either both agents $i$ and $i+1$ are assigned no goods, or agent $i$ is assigned strictly more goods than agent $i+1$.
It is easy to see that agent $i+1$ can never receive strictly more goods than agent $i$, because $w_{i+1} < w_i$, and swapping their bundles in such a situation strictly increases welfare, contradicting local efficiency. Suppose that at some point agent $i+1$ is assigned an item $j$, and at that moment both agents $i$ and $i+1$ have exactly $k$ goods. Let $X_i$ and $X_{i+1}$ be their respective bundles at that time. Then we have
$\sum_{j \in X_i} q_j \leq k \cdot (1 - \epsilon) + \sum_{j' < j} 2^{j'-n} \cdot \epsilon < k \cdot (1 - \epsilon) + 2^{j-n} \cdot \epsilon$,
and
$\sum_{j \in X_{i+1}} q_j \geq k \cdot (1 - \epsilon) + 2^{j-n} \cdot \epsilon$.
Hence, swapping the bundles of agents $i$ and $i+1$ increases the social welfare, violating local efficiency. This proves the claim.

We show that, following from the claim, agent $1$ is assigned at least $n$ goods. Suppose for a contradiction that agent $1$ receives at most $n - 1$ goods. Then, by induction, agent $i$ receives at most $\max(n - i,0)$ goods for every $i$. This implies that the algorithm allocates at most $n(n - 1)/2 < m$ goods, a contradiction. Therefore, agent $1$ receives at least $n$ goods.

Furthermore, we can inductively argue that each agent $i$ receives at most $n - i + 1$ goods. Thus agent $i$’s envy towards agent $1$ is at least $w_i \cdot (i - 1 - n \cdot \epsilon) \geq (1 - n \cdot \epsilon) \cdot (i - 1 - n \cdot \epsilon) \geq i - 2n \cdot \epsilon$. Consequently, the total subsidy is at least $2 + \ldots + n - 2n^2 \cdot \epsilon$.
\end{proof}

\subsubsection{Restricted Additive Valuations}
An instance with valuations $v_1, \ldots, v_n$ is \emph{restricted additive} if there are some $u_1, \ldots, u_m \geq 0$ such that for all $i \in [n]$, $v_i$ is additive and $v_i(\{g\}) \in \{0, u_g\}$ for $g \in [m]$.
Restricted additive valuations capture scenarios where all agents agree on the relative value of items (represented by $u_g$), but each agent is only interested in a subset of them; see \cite[e.g.][]{AkramiRS22}. An important special case of restricted additive valuations is the binary additive case, where $u_g = 1$ for all $g$; see \cite[e.g.][]{HalpernPPS20}.

\begin{theorem}\label{thm:restrictedadditive}
There exists an online algorithm that achieves envy-freeness with a total subsidy of at most $n(n-1)/2$ for all instances with restricted additive valuations.  Moreover, for any online algorithm, there exists an instance with binary additive valuations for which achieving envy-freeness requires a total subsidy of at least $n(n-1)/2$.
\end{theorem}
\begin{proof}
\newcommand{\circled}[1]{%
  \makebox[\linewidth][c]{%
    \tikz[baseline=(char.base)]{%
      \node[draw,circle,inner sep=1pt,
            text height=1.6ex,text depth=.4ex,
            outer sep=0pt] (char) {#1};}}}

\begin{table}[t]
\centering
\label{tab:agents-phases}
\begin{tabular}{|l||*{6}{P{0.035\linewidth}|}|P{0.035\linewidth}||*{2}{P{0.035\linewidth}|}|P{0.035\linewidth}|}
\hhline{|-||------| |-| |--| |-|}
\multicolumn{1}{|c||}{\textbf{Phase}}
& \multicolumn{6}{c||}{\textbf{P1}} & \textbf{P2} & \multicolumn{2}{c||}{\textbf{P3}} & \textbf{P4} \\
\hhline{|-||------| |-| |--| |-|}
\multicolumn{1}{|c||}{{Item}}
 & 1 & 2 & 3 & 4 & 5 & 6 & 7 & 8 & 9 & 10 \\
\hhline{|-||------| |-| |--| |-|}
Agent 1 & \circled{\cellcolor{green!20}1} & \circled{\cellcolor{green!20}1} & \cellcolor{green!20}1 & \cellcolor{green!20}1 & {\cellcolor{green!20}1} & \cellcolor{green!20}1 & {\cellcolor{green!20}1} & \circled{\cellcolor{green!20}1} & \cellcolor{green!20}1 & \circled{\cellcolor{green!20}1} \\
\hhline{|-||------| |-| |--| |-|}
Agent 2 & \cellcolor{green!20}1 & {\cellcolor{green!20}1} & \circled{\cellcolor{green!20}1} & \circled{\cellcolor{green!20}1} & \cellcolor{green!20}1 & {\cellcolor{green!20}1} & \cellcolor{green!20}1 & {\cellcolor{green!20}1} & \circled{\cellcolor{green!20}1} & \cellcolor{green!20}1 \\
\hhline{|-||------| |-| |--| |-|}
Agent 3 & \cellcolor{green!20}1 & \cellcolor{green!20}1 & {\cellcolor{green!20}1} & \cellcolor{green!20}1 & \circled{\cellcolor{green!20}1} & \cellcolor{green!20}1 & \circled{\cellcolor{green!20}1} & \cellcolor{green!20}1 & \cellcolor{green!20}1 & \cellcolor{red!20}0 \\
\hhline{|-||------| |-| |--| |-|}
Agent 4 & \cellcolor{green!20}1 & \cellcolor{green!20}1 & \cellcolor{green!20}1 & {\cellcolor{green!20}1} & \cellcolor{green!20}1 & \circled{\cellcolor{green!20}1} & \cellcolor{green!20}1 & \cellcolor{red!20}0 & \cellcolor{red!20}0 & \cellcolor{red!20}0 \\
\hhline{|-||------| |-| |--| |-|}
Agent 5 & \cellcolor{green!20}1 & \cellcolor{green!20}1 & \cellcolor{green!20}1 & \cellcolor{green!20}1 & \cellcolor{green!20}1 & \cellcolor{green!20}1 & \cellcolor{red!20}0 & \cellcolor{red!20}0 & \cellcolor{red!20}0 & \cellcolor{red!20}0 \\
\hhline{|-||------| |-| |--| |-|}
\end{tabular}
\caption{Illustration of the lower bound construction for binary additive valuations in the proof of Theorem~\ref{thm:restrictedadditive}. Red entries represent agent–item pairs with value 0, green entries represent pairs with value 1, and circled entries represent the allocation made by the online algorithm.}
\end{table}

\emph{Upper bound.}
Consider the following greedy algorithm. Let $(X_1, \ldots, X_n)$ denote the current allocation and $g$ be the incoming item. The algorithm assigns $g$ to the agent whose current bundle has the smallest total value among agents who value $g$ positively, that is, to agent $i \in \argmin_{a : v_a({g}) > 0} v_a(X_a)$. Because each item is allocated to an agent who values it positively and no other agent values it higher, the resulting allocation maximizes  welfare.

Let $X = (X_1, \ldots, X_n)$ be the final allocation. Without loss of generality, we can relabel the agents so that their final bundle values are in nonincreasing order: $v_1(X_1) \geq v_2(X_2) \geq \cdots \geq v_n(X_n)$.

We first show that for any pair of agents $i < j$, agent $i$ does not envy agent $j$. Since agent $j$ only receives items that they value positively, we have $v_j(X_j) \geq v_i(X_j)$. Together with our ordering assumption $v_i(X_i) \geq v_j(X_j)$, it follows that  $i$ does not envy $j$.

Next, consider any pair $j < i$. We show that agent $i$ envies agent $j$ by at most 1. Let $g$ be the last item assigned to $j$ that $i$ values positively, and let $Y$ denote the partial allocation just before $g$ arrives. At that time, the algorithm chose to give $g$ to $j$ rather than to $i$, so by definition of the greedy rule, $v_i(Y_i) \geq v_j(Y_j)$. Because agent $j$ values all items in $X_j$ positively, we also have $v_j(Y_j) \geq v_i(Y_j)$. Combining these inequalities gives $v_i(Y_i) \geq v_i(Y_j)$. When $g$ is added to $j$’s bundle, agent $i$’s value for $g$ increases $v_i(Y_j)$ by $v_i(\{g\}) \leq 1$, hence $v_i(X_i) \geq v_i(Y_i) \geq v_i(Y_j) + v_i(\{g\}) - 1 = v_i(X_j) - 1$. Therefore, agent $i$’s envy toward agent $j$ is at most 1.

By the characterization of minimum subsidies (Lemma~\ref{lem:additive-paths}), the subsidy required for agent $i$ is at most $i - 1$. Summing over all agents gives a total subsidy of $1 + 2 + \cdots + (n - 1) = n(n - 1)/2$.

\emph{Lower bound.} Fix an online algorithm for the class of restricted additive valuations. We adversarially construct an instance with binary additive valuations where this algorithm either requires  at least $\Omega(n^2)$ subsidy or violates envy-freeness. The adversary constructs the instance adaptively, based on the algorithm’s decisions, over a sequence of $n-1$ phases.

Throughout the execution, the adversary maintains a set of eliminated agents and eliminates a new agent at the end of each phase. The agent eliminated at the end of phase $p$ is denoted by $e_p$.

At the beginning of phase $p$, for $1 \leq p \leq n-1$, define
$C_p = \arg\min_{ i \in [n] \setminus \{e_1, \ldots, e_{p-1}\}} v_i(X_i)$, i.e., $C_p$ is the set of non-eliminated agents with minimum value. We let $C_p$ be the set of candidates for elimination. The adversary begins adding identical items, each of value $1$ for the non-eliminated agents in $[n] \setminus \{e_1, \ldots, e_{p-1}\}$ and value $0$ for the eliminated agents in $\{e_1, \ldots, e_{p-1}\}$. The adversary continues adding such items until one of the following occurs:
\begin{enumerate}[noitemsep,nosep]
    \item the algorithm assigns an item to an eliminated agent,
    \item $n^3$ items have been added during phase $p$, or
    \item all but exactly one candidate agent in $C_p$ have been assigned an item during phase $p$.
\end{enumerate}
At the end of phase $p$, the adversary selects a candidate agent $e_p \in C_p$ who received no item during phase $p$ to be eliminated.

\underline{Case 1.} The adversary terminates because an item is assigned to an eliminated agent. We show that in this case the algorithm violates local efficiency. Suppose the algorithm assigns an item to agent $e_q$ with $q < p$ during phase $p$. Let $X$ be the final allocation. Since every item except the last one was assigned to a non-eliminated agent at the time it arrived, the social welfare of $X$ is $\sw(X) = m-1$.

Pick a candidate $e_p$ that received no item in this phase. Consider an allocation $X'$ obtained from $X$ by reassigning $X'_{e_p} = X_{e_q}$ and $X'_{e_{p-j-1}} = X_{e_{p-j}}$ for $0 \leq j < p-q$. Note that agent $e_p$ values every item in $X_{e_q}$ at $1$, since $e_p$ was still active at the beginning of phase $p$. Similarly, for any $0 \leq j < p-q$, agent $e_{p-j-1}$ values all items in $X_{e_{p-j}}$ at $1$, because $e_{p-j}$ received no items during phase $p-j$ by the choice of elimination. Thus, $\sw(X') = m$, implying that $X$ is not locally efficient, and hence not envy-freeable.

\underline{Case 2.} The adversary terminates after $n^3$ items have arrived.
We show that the algorithm uses a subsidy of $\Omega(n^2)$.
By construction of valuations and the definition of $C_p$, at the beginning of phase $p$, each $i \in C_p$ weakly envies all other non-eliminated agents. Since $n^3$ items arrive, there  exists some non-eliminated agent $j \notin \{e_1, \ldots, e_{p-1}\}$ who receives at least $n^2$ items during phase $p$, and some candidate agent $i \in C_p$ who receives no item during that phase (otherwise we would be in case (3)).

Therefore, in the final allocation $X$, we have $v_i(X_i) + n^2 \leq v_i(X_j)$, implying that envy-freeness requires a subsidy of at least $n^2$.
    
\underline{Case 3.} The adversary proceeds to the next phase.
We show that at the end of phase $n-1$, the algorithm still requires a total subsidy of at least $\Omega(n^2)$.
By construction, for every $1 \leq p \leq n-2$, the agent $e_p$ eliminated in phase $p$ envies the agent $e_{p+1}$ eliminated in the next phase by at least $1$.
By the characterization of minimum subsidy (Lemma~\ref{lem:additive-paths}), this implies that the total subsidy is at least $1 + 2 + \ldots + (n-1) = n(n-1)/2$.
\end{proof}

\subsubsection{Identical Monotone Valuations}
An instance with valuations $v_1, \ldots, v_n$ is \emph{identical monotone} if, for all $i, j \in [n]$, $v_i(S) = v_j(S)$ for every bundle $S$. This class is quite general, imposing only monotonicity and normalization on the valuation, but it assumes that all agents have identical preferences; see \cite[e.g.][]{PlautR20}.

\begin{theorem}\label{thm:identical_monotone}
There exists an online algorithm that achieves envy-freeness with a total subsidy of at most $n-1$ for all instances with identical monotone valuations. Moreover, for any online algorithm, there exists an instance with identical monotone valuations for which achieving envy-freeness requires a total subsidy of at least $n-1$.
\end{theorem}

\begin{proof}
    \textit{Upper Bound.} The upper bound is obtained by allocating each incoming good to the agent with the lowest value (before the item arrived). This algorithm ensures that every pair of agents $i,j$ is envy free up to the removal of some good (i.e., either $i$ does not envy $j$ or there is an item in $j$'s bundle whose removal ensures this condition). To see this, let $X$ be the current allocation and suppose for a contradiction that there exist agents $i, j \in [n]$ such that $v(X_j \setminus \{g*\}) > v(X_i)$, where $g^* = \argmin_g v(X_j \setminus \{g\})$. Then, when the last good was allocated to agent $j$, agent $i$ had a lower value than agent $j$, contradicting the algorithm's decision to allocate the good to $j$\footnote{This is a classic algorithm for obtaining \textit{EF1} allocations, a popular fairness notion.}. Therefore, we have that the weight of every edge in the envy graph is upper bounded by $1$. Further, the envy graph can be represented by an ordering of the agents $i_1, \ldots, i_n$ such that for any $\ell < r$, $i_\ell$ has a positive envy towards $i_r$, and $i_r$ has an equal negative envy towards $i_\ell$. Now, consider the heaviest path in the resulting envy graph starting (without loss of generality) at agent $i_1$ and ending (without loss of generality) at agent $i_n$. We have that $v(X_{i_n}) = v(X_{i_{n-1}}) + w_{G_X}(i_{n-1}, i_n) = \cdots = v(X_1) + \sum_{j \in [n-1]} w_{G_X}(i_j, i_{j+1})$. Since the allocation is EF1, $v(X_{i_n}) \leq v(X_{i_1}) + 1$, therefore $\sum_{j \in [n-1]}w_{G_X}(j-1,j) \leq 1$. This implies that the weight of the heaviest path is at most $1$ and therefore by Lemma \ref{lem:additive-paths}, the total subsidy is at most $(n-1)$.

    \textit{Lower Bound.} If there is a single good that arrives with value $1$, it can only be allocated to one agent. Every other agent receives a subsidy of $1$, giving a total subsidy of $(n-1)$.
\end{proof}

\section{Discussion} \label{sec:discussion}

We conclude with a discussion of broader implications and open questions arising from our results.

\paragraph{Comparison with Envy Minimization.} A closely related line of work in online fair division focuses on the envy minimization problem without the local efficiency requirement, where the goal is to allocate items so as to minimize the maximum envy between any pair of agents (see \cite[e.g.][]{AleksandrovAGW15,HalpernPVX25}). This problem has been studied extensively, and seems to have connections to subsidy minimization ({\em i.e.,} Lemma \ref{lem:additive-paths}). For instance, \citet*{BenadeKPPZ24} obtain an envy bound against an adaptive adversary of $\widetilde{O}(\sqrt{m/n})$ for additive valuations over sufficiently many items, and show that this bound is essentially optimal. Lemma \ref{lem:additive-paths} seems to suggest that a bound on envy should imply a bound on subsidy, however, this holds only when the algorithm also maintains local efficiency at every step. This latter requirement poses the main challenge in online subsidy minimization.

\paragraph{Online Fair Division of Chores.} We remark that most of our results extend to the fair division of chores via similar analyses and without much additional effort. In fair chore division, the valuation functions are \textit{non-increasing}, and local efficiency continues to characterize envy-freeability. The subsidy, however, is measured differently: agents assigned extra chores receive greater payments. Consequently, while some of our subsidy bounds change with respect to their dependence on $n$, the boundaries we established for our two central questions remain unchanged.

\paragraph{Open Problems.} Our work creates some compelling directions for future research. Can the results we obtain for Questions 1 and 2 be extended to any other interesting valuation classes? Does there exist a meaningful characterization of the valuation classes below the boundary defined by either of the two questions? Another problem we leave open is determining the tight bound on the minimum subsidy for $k$-valued instances.

\bibliographystyle{ACM-Reference-Format} 
\bibliography{refs}

\end{document}